\documentclass[conference]{IEEEtran}
\IEEEoverridecommandlockouts
\usepackage{cite}
\usepackage{amsmath,amssymb,amsfonts}
\usepackage{amsthm}
\usepackage{algorithmic}
\usepackage{graphicx}
\usepackage{textcomp}
\usepackage{xcolor}
\usepackage{listings}
\usepackage{tikz}
\usepackage{subfig}
\usepackage{wrapfig}
\usepackage{rotating}
\usepackage{multirow}
\usepackage{listings}
\usepackage{placeins}
\usepackage{xspace}
\usetikzlibrary{matrix}
\usepackage{comment}
\usepackage{siunitx}
\usetikzlibrary{positioning, fit}
\usetikzlibrary{arrows,shapes,patterns,calc}
\def\BibTeX{{\rm B\kern-.05em{\sc i\kern-.025em b}\kern-.08em
    T\kern-.1667em\lower.7ex\hbox{E}\kern-.125emX}}

\tikzstyle{trans}=[font=\footnotesize]
\tikzstyle{single_s}=[rectangle,draw, double, trans, minimum size=8mm]
\tikzstyle{single_s2}=[rectangle,draw, double, trans, minimum 
width =12mm, minimum height=6mm]
\tikzstyle{single_is}=[rectangle,draw, double, fill=yellow, trans, minimum 
size=8mm]
\tikzstyle{kwc_s}=[rectangle,draw, trans, minimum 
width =18mm, minimum height=8mm]
\tikzstyle{kwc_is}=[rectangle,draw, trans, fill=yellow, minimum 
width =18mm, minimum height=8mm]
\tikzstyle{kwc_cs}=[rectangle,draw, trans, minimum 
width =12mm, minimum height=6mm]
\tikzstyle{state}=[rectangle,double,draw,trans, minimum size=8mm]
\tikzstyle{istate}=[circle,draw, trans, minimum size=8mm]
\tikzstyle{is_state2}=[rectangle,dotted,draw,trans, minimum 
size=8mm]
\tikzstyle{is_state}=[rectangle, fill=gray,dotted,draw,trans, minimum 
size=8mm]
\tikzstyle{rstate}=[rectangle,draw,trans, minimum size=10mm, 
font=\footnotesize, text width=21mm, align=center]
\tikzstyle{rstate2}=[rectangle,draw,trans, minimum size=10mm, 
text width=10mm, align=center]

\def\T{\hbox to 1.5em{\hfill}}

\def\L#1{\raise .2ex\hbox{\scriptsize{$#1$}}&}

\newtheorem{coro}{Corollary}
\newtheorem{theorem}{Theorem}
\newtheorem{definition}{Definition}
\newtheorem{lemma}{Lemma}

\newcommand{\certifier}{\textsc{Certifaiger++}\xspace}
\newcommand{\wlcertifier}{\textsc{Certifaiger-wl}\xspace}

\newcommand{\vars}{\ensuremath\mathit{vars}}

\begin{document}

\title{Stratified Certification for $k$-Induction\\\LARGE{(Extended Version)}}

\author{\IEEEauthorblockN{Emily Yu\IEEEauthorrefmark{1},
		Nils Froleyks\IEEEauthorrefmark{1}, Armin 
		Biere\IEEEauthorrefmark{2} and
		Keijo Heljanko\IEEEauthorrefmark{3}\IEEEauthorrefmark{4}}\\
\IEEEauthorrefmark{1}Johannes Kepler University, Linz, Austria\\
	\IEEEauthorrefmark{2}Albert–Ludwigs–University,	Freiburg, Germany\\
	\IEEEauthorrefmark{3}University of Helsinki, Helsinki, Finland\\
	\IEEEauthorrefmark{4}Helsinki Institute for Information Technology, 
	Helsinki, 
	Finland}

\maketitle

\begin{abstract}
Our recently proposed certification framework for bit-level 
$k$-induction-based model 
checking has been shown to be quite effective 
in increasing the trust of verification results
even though
it partially involved quantifier reasoning.
In this paper we show how to simplify the approach
by assuming reset functions to be stratified.
This way it can be lifted to word-level and in principle to
other theories where quantifier reasoning is difficult.
Our new method requires six simple SAT checks and one 
polynomial-time check, allowing certification to remain in co-NP while
the previous approach required five SAT checks and one QBF check.
Experimental results show
a substantial performance gain
for our new approach.
Finally we present and evaluate our new tool \wlcertifier which is able to
certify $k$-induction-based word-level model checking.

\end{abstract}

\section{Introduction}

Over the past several years, there has been growing interest in system 
verification 
using word-level reasoning. Satisfiability Modulo 
Theories (SMT) solvers for the theory of fixed-size 
bit-vectors are widely used for word-level 
reasoning~\cite{DBLP:conf/cav/NiemetzPB16,DBLP:journals/fmsd/NiemetzPB17}.
 For example, 
word-level model checking has been an 
important 
part of the hardware model checking competitions since 2019. Given the 
theoretical 
and practical importance of word-level verification, a generic certification 
framework for it is necessary. As quantifiers in combination with 
bit-vectors are challenging for SMT solvers and various 
works have focused on eliminating quantifiers in 
SMT~\cite{DBLP:journals/jar/NiemetzP0ZBT21,
	DBLP:journals/fmsd/NiemetzPB17,DBLP:conf/cav/NiemetzPRBT18}, 
	a main 
	goal of this paper is to 
generate certificates without quantification which are thus efficiently 
machine 
checkable. 

	Temporal induction (also known as 
	$k$-induction)~\cite{DBLP:conf/fmcad/SheeranSS00} is a 
well-known model checking technique for verifying software and hardware 
systems. An attractive feature of $k$-induction is that it is natural to 
integrate it with modern SAT/SMT solvers, making it popular in both 
bit-level model checking and beyond~\cite{ChampionMST16,DBLP:conf/cav/MouraORRSST04,JovanovicD16},
including word-level model checking.

Certification helps gaining 
confidence in model checking results, 
which 
is important for both safety- and business-critical applications. There have 
been several contributions focusing on generating proofs for 
SAT-based model 
checking\cite{DBLP:conf/fm/ConchonMZ15,DBLP:conf/fmcad/GurfinkelI17,DBLP:conf/hvc/KuisminH13,DBLP:conf/cav/Namjoshi01,DBLP:conf/nfm/WagnerMTCS17,DBLP:conf/fmcad/GriggioRT18,DBLP:conf/icfem/YuBH19}.
 For example~\cite{DBLP:journals/fmsd/GriggioRT21} 
and~\cite{DBLP:conf/fmcad/GriggioRT18} proposed an approach to certify 
LTL properties and a few preprocessing techniques by generating deductive 
proofs. In 
this paper, we focus on finding an inductive invariant for 
$k$-induction. Unlike 
other SAT/SMT-based model checking techniques such as 
IC3~\cite{DBLP:conf/vmcai/Bradley11} and 
interpolation~\cite{DBLP:conf/cav/McMillan03,DBLP:journals/tcs/McMillan05},
$k$-induction does not 
automatically 
generate an inductive invariant that can be used as a 
certificate~\cite{DBLP:books/daglib/0080029}. In 
previous research~\cite{DBLP:conf/cav/YuBH20}, 
certification of $k$-induction 
can be achieved via five SAT checks together with a one-alternation QBF 
check, 
redirecting the certification problem to verifying an inductive invariant in 
an extended model that 
simulates the original one.

At the heart of the present contribution is the idea of reducing the 
certification method of 
$k$-induction to pure SAT checks, \textit{i.e.,} eliminating the quantifiers. 
This enables us to complete the 
certification procedure at a lower complexity, and to directly apply the 
framework 
to word-level certification. We introduce 
the notion of stratified 
simulation which allows us to reason about the simulation relation between 
two systems.

This 
stratified simulation relation can be verified by three SAT and a 
polynomial-time check. The latter checks whether the reset function definition is 
indeed stratified. In 
addition to that, we present a witness circuit 
construction which simulates the original circuit under the stratified 
simulation relation thus creating a simpler and more elegant certification 
construction for $k$-induction. 

While the previous work only focused on bit-level model checking, we also 
lift our method to word-level by implementing a complete toolsuite 
\wlcertifier, where the experiments show the practicality and effectiveness 
of 
our 
certification method for word-level models.

The rest of the paper is organised as follows. 
Section~\ref{bg} includes the
preliminaries where we fix the notation and basic concepts. In 
Section~\ref{certification} we present the formal definition of stratified 
simulation and the witness circuit 
construction for $k$-induction with an inductive invariant. We continue to 
show that the 
witness circuit simulates the given circuit and provide a formal proof that 
the inductive invariant can be used as a 
proof certificate. We give details of our implementation of the approach for 
both bit-level and word-level model checking in 
Section~\ref{exp}, where we also provide experimental results of our 
toolkits on different sets of benchmarks.

\section{Background}~\label{bg}

This paper extends previous work in certification for $k$-induction-based 
bit-level model 
checking~\cite{DBLP:conf/cav/YuBH20}. In this section, we present 
essential
concepts and notations.

For the sake of
simplicity we work with \emph{functions} represented as interpreted terms and formulas over
fixed but arbitrary theories which include an equality predicate.
We further assume a finite sorted set of variables $L$ where each variable $l\in L$ is 
associated with a finite domain of possible values.
We also include Boolean 
variables as variables with a domain of 
$\{\top, \bot\}$, for which we keep standard notations.

For two 
sets of 
variables $I$ and $L$, we also write $I,L$ to 
denote their union. Given two
functions $f(V)$, $g(V')$ where $V\subseteq V'$ (represented as interpreted terms
over our fixed but arbitrary theories)
we call them \emph{equivalent}, written $f(V) \equiv g(V')$, if for every 
assignment to variables in $V$ and $V'$ that matches on the shared set of 
variables $V$, the functions $f(V)$, $g(V')$ have the same values. Additionally, we 
use ``\,$\simeq$\," for syntactic 
equivalence~\cite{DBLP:books/el/RV01/DegtyarevV01a}, ``$\,\rightarrow\,$" 
for syntactic implication, and 
``\,$\Rightarrow$\," for semantic 
implication.  To define semantical
concepts or abbreviations we stick to equality ``$\,=\,$".
We use $\vars(f)$ to denote the set of variables occurring in the
syntactic representation of a function $f$.

In word-level model checking operations are applied to fixed-size 
bit-vectors. We introduce the notion of word-level circuits where we model 
inputs and latches as finite-domain variables. The 
specification of a circuit is described in terms of a ``good states" property
(so we focus on safety in this paper).

\begin{definition}[Circuit]\label{circuit}
	A circuit is a tuple $C=(I, L, R, F, P)$ such that:
	\begin{itemize}
		\item $I$ is a finite set of \emph{input} variables.
		\item $L$ is a finite set of \emph{latch} variables.
		\item $R= \{ r_l(L) \mid l \in L\}$ is a set of 
		\emph{reset 
			functions}.
		\item $F=\{ f_l(I,L) \mid l \in L\}$ is a set of
		\emph{transition functions}.
		\item $P(I,L)$ is a function that evaluates to a Boolean output,
		encoding the \emph{(good states) property}.
	\end{itemize}
\end{definition}

\noindent
By Def.~\ref{circuit} a circuit represents a hardware system in a fully 
symbolic form. It is very close to bit-level models encoded in 
AIGER~\cite{Biere-FMV-TR-11-2} format 
as well as word-level models in
BTOR2~\cite{DBLP:conf/cav/NiemetzPWB18} format - both used in 
hardware model checking competitions.
A state in the system is an assignment to the latches, and any state where 
all latches match their reset functions is an initial state. In order to talk 
about the reset functions of a subset of latches $L''\subseteq L$, 
we also write 
\[R(L'')=\bigwedge\limits_{l\in {L''}} (l\simeq r_l(L)).\] 
The inputs 
represent the 
environment, which can be assigned any value thus bringing 
non-determinism to the circuit behaviour.
The following four definitions are adapted from our previous work \cite{DBLP:conf/cav/YuBH20}
for completeness of exposition.

\begin{definition}[Unrolling]\label{def-unrolling}
	For an unrolling depth $m\in \mathbb{N}$, the \emph{unrolling} of a 
	circuit $C=(I, 
	L, R, F, P)$ of 
	length 
	$m$ is 
	defined as $U_m=   \bigwedge\limits_{i\in [0,m)} 
	(L_{i+1}\simeq 
	F(I_i, 
	L_i)) $.
\end{definition}

\noindent
For clarity, we use subscript such as $L_i$ for the temporal direction 
(\textit{e.g.,} in the context of unrolling), and 
later in the paper we use superscript for 
the spatial direction.

	\begin{definition}[Inductive invariant]\label{ind}
	Given a circuit $C$ with a property $P$, $\phi(I,L)$ is an
	\emph{inductive invariant} in 
	$C$ if and only if the following conditions hold: 
	
	\begin{enumerate}
		\item $ R(L)\Rightarrow \phi(I,L)$,
		\hfill``initiation''~~
		\item $\phi(I,L) \Rightarrow P(I,L)$, and
		\hfill``consistency''~~
		\item $U_1 \wedge \phi(I_0, L_0) \Rightarrow \phi(I_1, L_1).$
		\hfill``consecution''~~
	\end{enumerate}
	
\end{definition}

\noindent
As a generalisation of the notion of an inductive invariant, $k$-induction 
checks $k$ steps of unrolling instead of 
1. In the following, to verify that a property is an inductive invariant, 
we consider it as the special case of $k$-induction with $k=1$ and 
$\phi(I,L) 
= P(I,L)$.

	\begin{definition}[$k$-induction]\label{kind}
	Given a circuit $C$ with a property $P$, $P$ is called
	\emph{$k$-inductive} in 
	$C$ if and only if the following two conditions hold: 
	
	\begin{enumerate}
		\item\label{bmc} $ U_{k-1} 
		\wedge 
		R(L_0)
		\Rightarrow \bigwedge\limits_{i\in 
			[0,k)}P(I_i, 
		L_i)$, and
		\hfill``BMC''~~
		\item\label{cons} $U_k\wedge \bigwedge\limits_{i\in 
			[0,k)}P(I_i, 
		L_i)\Rightarrow 
		P(I_k, L_k)$.
		\hfill``consecution''~~
	\end{enumerate}

\end{definition}

\begin{definition}[Combinational extension]\label{ext} \leavevmode\\
	A circuit $C'=(I', L', R', F', P')$ combinationally extends a circuit $C=(I, 
	L, R, F, P)$ \;if\; $I=I'$ and $L\subseteq L'.$
\end{definition}

\section{Certification}~\label{certification}

In this section we introduce and formalise our certification approach which 
reduces 
the certification problem to six SAT checks and one polynomial 
stratification check. 

The certification approach is outlined in 
Fig.~\ref{outline}. 
Intuitively, a \emph{witness circuit} is generated from a 
given 
value of $k$ (provided by the model checker) and a model (either 
bit-level or word-level). The witness 
circuit simulates the original circuit while allowing more behaviours (we 
formally define 
it as the \textit{stratified simulation relation}). In practice, the witness 
circuit would be required to be provided by model 
checkers as the certificate in hardware 
model checking competitions.
To verify 
the stratified simulation relation between the two circuits, three conditions 
are generated as SAT formulas (listed in Def.~\ref{sim}).

We also perform a 
polynomial-time 
stratification 
check on the witness circuit. The check requires that the definition of the 
reset function is stratified, \textit{i.e.}, no cyclic dependencies between the 
reset definitions of the variables exist. This is the case
for all hardware model checking competition benchmarks. Even though cyclic
definitions 
have been the subject of study in several
papers~\cite{DBLP:journals/tcad/Malik94,DBLP:conf/iccad/JiangMB04,riedel2004cyclic},
 they are
 usually avoided due
to the complexity of their analysis and subtle effects on semantics.

The 
approach 
in \cite{DBLP:conf/cav/YuBH20} can handle cyclic resets but at the cost of 
QBF quantification, and 
thus \cite{DBLP:conf/cav/YuBH20} not being able to be efficiently adapted 
to the context of 
word-level verification. Furthermore, the witness circuit includes an 
inductive invariant which serves as a proof certificate, which is verified by 
another three SAT checks as defined in Def.~\ref{ind} ($\varphi_{init}, 
\varphi_{cons}, \varphi_{consec}$). Certification is successful if all seven 
checks 
pass.

We 
begin by defining 
stratified reset 
functions.

	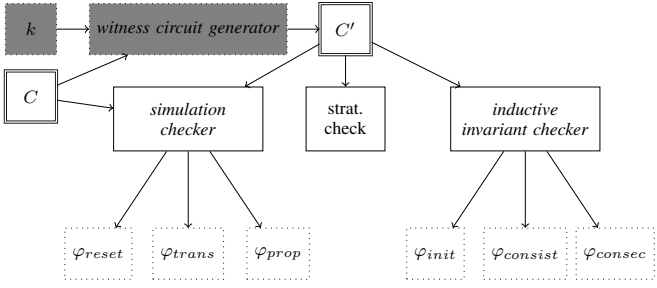
\begin{figure}
	\centering
	\scalebox{0.85}{
		\begin{tikzpicture}
		[node distance = 5em]
		\node [is_state]
		[]
		(S70)
		{$k$};
		\node [state]
		[below of=S70, node distance = 3em]
		(S0)
		{$C$};
		
		\node [is_state]
		[right of=S70, node distance = 7em]
		(S1)
		{ \textit{witness circuit generator}};
		\node [state]
		[right of=S1, node distance = 7em]
		(S2)
		{$C'$};
		
		\node [rstate]
		[below of=S1, node distance = 4em]
		(S4)
		{\textit{simulation checker}};
		\node [rstate]
		[right of=S4, node distance = 15em]
		(S3)
		{\textit{inductive invariant checker}};
		\node [is_state2]
		[below of=S4, node distance = 6em]
		(S5)
		{$\varphi_{trans}$};
		\node [is_state2]
		[left of=S5, node distance = 4em]
		(S6)
		{$\varphi_{reset}$};
		\node [is_state2]
		[right of=S5, node distance = 4em]
		(S7)
		{$\varphi_{prop}$};	
		\node [is_state2]
		[below of=S3, node distance = 6em]
		(S8)
		{$\varphi_{consist}$};	
		\node [is_state2]
		[left of=S8, node distance = 4em]
		(S9)
		{$\varphi_{init}$};	
		
		\node [is_state2]
		[right of=S8, node distance = 4em]
		(S10)
		{$\varphi_{consec}$};	
		\node [rstate2]
		[below of=S2, node distance = 4em]
		(S71)
		{strat. check};	
		
		\draw[->]
		(S0) to [] node[left] {} 
		(S1);
		\draw[->]
		(S70) to [] node[left] {} 
		(S1);
		\draw[->]
		(S1) to [] node[left] {} 
		(S2);
		\draw[->]
		(S0) to [] node[left] {} 
		(S4);
		\draw[->]
		(S2) to [] node[left] {} 
		(S4);
		\draw[->]
		(S2) to [] node[left] {} 
		(S3);
		\draw[->]
		(S4) to [] node[left] {} 
		(S6);
		\draw[->]
		(S4) to [] node[left] {} 
		(S7);
		\draw[->]
		(S3) to [] node[left] {} 
		(S8);
		\draw[->]
		(S3) to [] node[left] {} 
		(S9);
		\draw[->]
		(S3) to [] node[left] {} 
		(S10);
		\draw[->]
		(S4) to [] node[left] {} 
		(S5);
		\draw[->]
		(S2) to [] node[left] {} 
		(S71);

		\end{tikzpicture}
	}
	\caption{\label{outline} 
		An outline of the certification approach. Given some value of $k$ and 
		a model $C$, $C'$ is the resulting witness circuit. The coloured area is 
		specific to our approach for $k$-induction, and the rest corresponds 
		to the general certification flow.}
\end{figure}

\begin{definition}{(Dependency graph.)}
	Given a set of latches $L$ and a set of reset functions $R=\{r_l \mid l\in 
	L\}$, 
	the dependency graph 
	$G_R$ has latch variables $L$ as nodes and
	contains a directed edge $(a,b)$ from $a$ to $b$ iff $a\in \vars(r_b)$ 
	and $r_b \neq b$.
\end{definition}

\noindent
Latches with undefined reset value are common in applications.
We simply set $r_b = b$ for some uninitialised latch $b$
in such a case
(as in AIGER and BTOR) to avoid being required to reason about
ternary logic or partial functions.
Thus the syntactic condition ``$r_b \neq b$''
in the last definition simply avoids spurious self-loops in the dependency
graph for latches with undefined reset values.

\begin{definition}{(Stratified resets.)}~\label{stratreset}
	Given a set of latches $L$, and a set of reset functions $R=\{r_l \mid l\in 
	L\}$. 
	$R$ is said to be stratified iff $G_R$ is acyclic.
\end{definition}

\begin{definition}{(Stratified circuit.)}
	A circuit $C=(I, 
	L, R, F, P)$ is said to be stratified iff $R$ is stratified.
\end{definition}

The stratification check can be done in 
polynomial time using Def.~\ref{stratreset} and it is enforced 
syntactically in the two hardware description formats
AIGER and BTOR2.

\begin{definition}{(Stratified simulation.)}\label{sim}
	Given two stratified circuits $C$ and $C'$, where $C'$ combinationally 
	extends 
	$C$. There is a stratified simulation between $C'$ and $C$ iff,
\begin{enumerate}
	\item \label{def-cs-r}  $r_l(L)\equiv r'_l(L')$ for $l\in L$,
	\hfill``reset''~~
	\item \label{def-cs-f}	$f_l(I, L)\equiv f'_l(I,L')$ for $l\in L$, and
	\hfill``transition''~~
	\item \label{def-cs-p}  $P'(I,L')\Rightarrow
	P(I,L)$.
	\hfill``property''~~

\end{enumerate}
\end{definition}

\noindent
In essence, the crucial change here compared to the combinational 
simulation definition in ~\cite{DBLP:conf/cav/YuBH20} is the reset 
condition, whose simplification was possible under the stratification 
assumption. The 
above three conditions are encoded into SAT/SMT formulas 
($\varphi_{reset}, \varphi_{trans}, \varphi_{prop}$ in Fig.~\ref{outline}) 
which are then checked by a solver for validity. In the rest of the 
paper, we simply refer to the stratified simulation relation 
as 
simulation relation.

\begin{theorem}
Given two circuits $C=(I,L,R,F,P)$ and $C'=(I',L',R',F',P')$, where $C'$ 
simulates $C$. If $C'$ is safe, then $C$ is also safe. 
\end{theorem}
\begin{proof}
	We assume $C'$ is safe such that $R'(L')\wedge U'_i\Rightarrow 
	\bigwedge\limits_{j\in[0,i]}P'(I'_j,L'_j)$ holds for all $i\geq 0$. To show 
	that $C$ 
	is safe, 
	we do a proof by contradiction by assuming the safety check in $C$ fails 
	for 
	some  $m$ such that $R(L)\wedge U_m \wedge \neg P(I_m,L_m)$ has a satisfying 
	assignment $s$ and we fix this assignment. We can then construct a satisfying 
	assignment for $R(L)\wedge R'(L')$ (according to Def.~\ref{sim}) as under the 
	stratification assumption for resets, $R'(L')$ is 
	guaranteed to be satisfiable with $L_0$ being a subset of $L'_0$. This is 
	the case as the stratification of the reset function $R'$ does not allow 
	any 
	cyclic dependencies that could potentially make $R'$ unsatisfiable. 
	Moreover, by 
	the transition check in Def.~\ref{sim}, the unrolling 
	$U'_m$ of $C'$ is also satisfiable with the transition function $F$ applied on 
	the projected (or shared) components of both circuits.
	Since the rest of the latches in $L'$ follow a transition function, they are 
	also satisfiable 
	(transition functions 
	guarantee that there is always a successor state for all states). Therefore 
	the 
	initialised 
	unrolling $R'(L'_0)\wedge U'_m$ is satisfiable. Furthermore, by our 
	assumption, 
	$\bigwedge\limits_{i\in[0,m]}P'(I'_i,L'_i)$ holds. By 
	Def.~\ref{ext},         
	projected latches of $C'$ stay the same as $L_i$ for all
	$i\in[0,m]$; thus by
	Def.~\ref{sim} we have that
	$\bigwedge\limits_{i\in[0,m]}P(I_i,L_i)$ holds. Therefore we reach a 
	contradiction.
\end{proof}

Next, we introduce the witness circuit construction. This is similar to the 
construction in ~\cite{DBLP:conf/cav/YuBH20} but differs in several details, 
\textit{e.g.},  the reset function definition is stratified and significantly 
simplified compared to \cite{DBLP:conf/cav/YuBH20}.

\begin{definition}{(Witness circuit.)}\label{wit}
	Given a circuit $C=(I,L,R,F,P)$ and an integer $k\in \mathbb{N^+}$, its 
	witness 
	circuit 
	$C'=(I',L',R',F',P')$ is defined 
	as follows:
	\begin{enumerate}
		\item $I'=I $ (also referred to as $X^{k-1}$),
		\item $L'= L^{k-1} \cup\cdots\cup L^0 \cup 
		X^{k-2}\cup\cdots\cup 
		X^0
		\cup 
		 B$ where,
		\begin{itemize}
		\item $L^{k-1}=L,$ the other variables sets are copies of $I$ 
		and $L$ 
		respectively with the same variable domains.
		\item $B=\{b^{k-1},\cdots, 
		b^0\} $ are Booleans.
		\end{itemize}
		
		\item $R': $ \label{reset}
		\begin{itemize}
			\item for $l\in L^{k-1}, r'_l=r_l(L^{k-1}).$ 
			\item for $l\in L^0\cup \cdots\cup L^{k-2}\cup X^0\cup \cdots\cup 
			X^{k-2}, r'_l=l$.
			\item $r'_{b^{k-1}}=\top$.
			\item for $i\in [0,k-1), r'_{b^i}=\bot.$
		\end{itemize}
		\item $F':$\label{transition}
		\begin{itemize}
			\item for $l\in L^{k-1}, f'_l=f_l(I',L^{k-1}).$ 
			\item $f'_{b^{k-1}}=b^{k-1}.$
			\item for $i\in [0,k-1), l^i\in (L^i\cup X^i \cup \{b^i\}), f'_{l^i}=l^{i+1}$.
		\end{itemize}
		\item $P'=\bigwedge\limits_{i\in[0,4]} p_i(I',L')$ where
		\begin{itemize}
				
{
	\iftrue
	\def\LEFT#1{\hbox to 4em{\hfill$#1$}}
	\def\RIGHT#1{\kern-.25em\hbox to 3.5em{\hfill$#1$}}
	\def\RANGE#1{\kern-.25em\hbox to 3.5em{\hss$\scriptstyle #1$\hss}}
	\else
	\def\LEFT#1{#1}
	\def\RIGHT#1{#1}
	\def\RANGE#1{#1}
	\fi
	
	\smallskip
	\item $\LEFT{p_0(I',L')}=  \bigwedge\limits_{\RANGE{i\in[0,k-1)}}
	(b^{i} \rightarrow b^{i+1}) $.
	\item $\LEFT{p_1(I',L')}= \bigwedge
	\limits_{\RANGE{i\in[0,k-1)}}
	(b^i\rightarrow  (L^{i+1} \simeq F(X^{i}, L^{i})))$.
	\item $\LEFT{p_2(I',L')}= \bigwedge\limits_{\RANGE{i\in [0,k)}} (b^i 
	\rightarrow P(X^i,
	L^i))$.
	\item $\LEFT{p_3(I',L')} = 
	\bigwedge\limits_{\RANGE{i\in[1,k)}}((\neg b^{i-1}\wedge b^i)\rightarrow 
	R(L^i))$.
	\\[.25ex]
	\item $\LEFT{p_4(I',L')} = \RIGHT{b^{k-1}.}$
}

\end{itemize}
\end{enumerate}
\end{definition}

Here we extend a given circuit to a witness circuit, which has $k$ copies 
of the original latches and inputs, and additional $k$ latches of $B$ that we 
refer to as the initialisation bits. We refer to the $\{k-1\}th$ as the most 
recent, and the $0$th as the oldest. Intuitively the most recent copy unrolls 
in the same way as the original circuit, with the older copies copying the 
previous values of the younger copies. When all initialisation 
bits are $\top$, we say the machine has reached a ``full initialisation" 
state. 

\begin{lemma}~\label{strat}
 Given a circuit $C$ with reset function $R$ and its witness circuit $C'$ 
 with reset function $R'$. If $R$ is stratified, then $R'$ is also stratified.
\end{lemma}
\begin{proof}
 We assume $R$ is stratified. By Def.~\ref{wit}.~\ref{reset}, the initialisation  
 bits in $B$ all have constant initial values, and the copies of latches and 
 inputs are uninitialised. Therefore they have no incoming edges in  
 $G_{R'}$ and  we can remove them without removing any cycles in  
 $G_{R'}$. The remaining latches $L^{k-1}$ have the same reset functions 
 as in $C$, 
 which by 
 our assumption, are stratified. Therefore $R'$ is stratified.
\end{proof}

We now prove the stratified simulation relation between $C'$ and $C$.

\begin{theorem}\label{thm-kwc-cs}
	Given a circuit $C$ and its witness circuit $C'$. $C'$  
	simulates $C$.
\end{theorem}
\begin{proof}
The inputs stay unchanged in the witness circuit, as well as the 
latches $L$  
(which are mapped to $L^{k-1}$), therefore $C'$ combinationally extends 
$C$. By Def.~\ref{wit}.\ref{reset}, for $l\in L^{k-1}, r'_l=r_l(L^{k-1})$ 
and by Lemma~\ref{strat} $R'$ is stratified thus satisfiable; therefore
the reset condition in Def.~\ref{sim} is fulfilled. By 
Def.~\ref{wit}.\ref{transition},  $ 
f'_l=f_l(I',L^{k-1})$ for 
$l\in L^{k-1}$. Thus the 
transition 
condition passes.  The property condition follows as $p_4$ and $p_2$ of 
$P'$ 
together imply $P$.
\end{proof}

Next we prove that under the assumption that $C'$ stratified simulates 
$C$, 
we 
have that $C$ is $k$-inductive iff $C'$ is 1-inductive. We proceed by first 
considering the BMC check.

\begin{lemma}~\label{lem2}
	If the BMC check for the unrolling of length $k-1$ passes in $C$, the 
	BMC check also passes for the unrolling of length $0$ in its witness 
	circuit $C'$.
\end{lemma}
\begin{proof}
Assume $U_{k-1}\wedge R(L_0)\Rightarrow \bigwedge\limits_{i\in 
	[0,k)}P(I_i,L_i)$ such that the initialised unrolling of $C$ is safe. We do a 
	proof by contradiction by assuming the BMC check fails in $C'$ such that 
	$R'(L'_0)\wedge\neg P'(I'_0,L'_0)$ has a satisfying assignment $s$. We 
	can then construct a satisfying asignment for $R(L_0)$ by using the 
	mapping $L_0=L^{k-1}_0$. By our assumption of the BMC check passing 
	in $C$, we have $P(I_0,L_0)$. According to the reset functions in 
	Def.~\ref{wit}, $b^{k-1}_0=\top$, and the rest of initialisation bits are set 
	to $\bot$. The premise of 
$p_2(I'_0,L'_0)$ is only satisfied for 
$b^{k-1}_0$, and with the same assignment satisfying 
$P(X^{k-1}_0,L^{k-1}_0)$, $p_2(I'_0,L'_0)$ is also 
satisfied.
Lastly, the premise of 
$p_3(I'_0,L'_0)$ is only satisfied for $\neg b^{k-2}_0\wedge b^{k-1}_0$, 
and since $R(L^{k-1}_0), p_3(I'_0,L'_0)$ is satisfied. Therefore we have 
$P'(I'_0,L'_0)$, and thus a contradiction is reached.
\end{proof}

Note that all that is proved here is that
if $C$ is safe for $k$ time steps then the circuit $C'$ is safe for one time 
step. 
This differs from the proof strategy in~\cite{DBLP:conf/cav/YuBH20} 
(Lemmas 3-5), and allows for the 
simpler reset definition in this paper.

\begin{lemma}\label{lem3}
	If the property $P$ is $k$-inductive in $C$, the consecution check for 
	the unrolling of length 1 passes 
	in its witness circuit $C'$.
\end{lemma}
\begin{proof}
	Based on our construction of the witness circuit, there are two possibilities of a 
	state: 1) $C'$ is partially initialised; or 2) all 
	initialisations bits are set. In the following proof, we consider the two 
	cases separately. In both 
	cases, we first assume $P$ is 
	$k$-inductive in $C$.
	
	We start by considering the case where $L'_0$ is partially initialised. Let $U'_1$ be 
	the unrolling of $C'$,  
	$m\in[1,k)$ be the lowest index such that 
	$b^0_0,...,b^{m-1}_0$ are set to $\bot$, while $b^m_0,...,b^{k-1}_0$ are 
	set 
	to $\top$, which follows from $p_0$ and $p_4$. We do a proof by 
	contradiction, and assume there is a satisfying assignment
	$s$ of the negation of the consecution check formula $U'_1\wedge 
	P'(I'_0,
	L'_0)\wedge \neg P'(I'_1,L'_1)$. Based on $p_3(I'_0,L'_0)$ we have 
	$R(L^m_0)$.
	We also have $L^{i+1}_0\simeq F(X^i_0, L^i_0)$ for $i\in[m,k-1)$, based on  
	$p_1(I'_0,L'_0)$. 
	Furthermore, $U'_1$ implies 
	$L'_1\simeq F'(I'_0,L'_0)$, and by the transition function construction in 
	Def.~\ref{wit}, 
	$L^{k-1}_1\simeq F(X^{k-1}_0,L^{k-1}_0)$. Therefore we can construct a 
	satisfying assignment for $U_{k-1} \wedge R(L_0)$ where  
	$I_{i-m}=X^i_0, 
	L_{i-m}=L^i_0$ for all $i\in[m,k)$, and $I_{k-m}=I'_1, 
	L_{k-m}=L^{k-1}_1$. The variables in $U_{k-1}$ that are not yet defined 
	have a satisfying assignment since we use transition functions. By our 
	assumption that the BMC check passes in $C$, we have 
	$P(X^i_0,L^i_0)$ for all 
	$i\in [m,k)$ and $P(I'_1, L^{k-1}_1)$. We then proceed to prove $P'(I_1', L_1')$ 
	is indeed satisfied. Similar to 
	our proof 
	in Lemma~\ref{lem2}, based on the definition of transition functions in Def.~\ref{wit}, 
	$b^{i}_1=\top$ for all $i\in[m,k)$ while $b^i_1=\bot$ for all $i\in[0,m)$. 
	Additionally, $X^{i}_1=X^{i+1}_0, L^i_1=L^{i+1}_0$ for $i\in[0,k)$. The 
	rest 
	of the proof follows the same logic as Lemma~\ref{lem2}, for showing 
	$P'(I_1', L_1')$ is satisfied. We then reach a contradiction here.
	
	Now we continue to prove the second scenario where all initialisation bits are 
	$\top$. Again, let $U'_1$ be the unrolling of $C'$ with 
	$B_0=\{b^0_0,...,b^{k-1}_0\}$ all set to 
	$\top$. We again do a proof by 
	contradiction assuming there is a satisfying assignment
	$s$ of $U'_1\wedge 
	P'(I'_0,
	L'_0)\wedge \neg P'(I'_1,L'_1)$. By the transition property $p_1(I'_0, 
	L'_0)$, 
	the 
	components follow the transition function $F$, such that $L^{i+1}_1\simeq 
	F(X^{i}_0,L^{i}_0)$ for all $i\in[0,k-1)$. Similar to our argument above, $U'_1$ 
	implies $L^{k-1}_1\simeq F(I'_0, 
	L^{k-1}_0)$. We also have $\bigwedge\limits_{i\in[0,k)}P(X^i_0,L^i_0)$ 
	based on $p_2(I'_0,L'_0)$ and the values of $B_0$ (\textit{i.e}., the values 
	of $B$ at timepoint 0). Thus we have a satisfying assignment for 
	$U_k\wedge 
	\bigwedge\limits_{i\in[0,k)}P(I_i,L_i)$ 
	with $L_i= L^{i}_0 
	\wedge I_i= 
	X^i_0$ for all $i\in[0,k)$ and $I_k = X^{k-1}_1, L_k=L^{k-1}_1$. Based 
	on our 
	assumption that the consecution of $C$ passes, we have 
	$P(I'_1,L^{k-1}_1)$. Following the same reasoning as in 
	the first case, after one transition, $b^i_1=\top$ for all 
	$i\in[0,k)$, and $X^{i}_1=X^{i+1}_0, L^i_1=L^{i+1}_0$ for $i\in[0,k-1)$. 
	We will now 
	show $P'(I_1', L_1')$ is satisfied. The sub-property
	$p_2(I'_1,L'_1)$ is satisfied as we have proved $P(X^i_1, L^i_1)$ for all 
	$i\in[0,k)$. As $U_k$ 
	is satisfied which implies $L^{i+1}_1\simeq F(X^i_1, 
	L^{i}_1)$ for $i\in[0,k)$, $p_1(I'_0,L'_0)$ is preserved. Based on the 
	values of $B_1$ (\textit{i.e}., the values of $B$ at timepoint 1), 
	$p_0(I'_1,L'_1),p_3(I'_1,L'_1),p_4(I'_1,L'_1)$ are satisfied immediately. We 
	conclude the $P'(I'_1, L'_1)$ 
	is 
	satisfied and we reach a contradiction. 
\end{proof}

Based on Lemma~\ref{lem2} and ~\ref{lem3}, we immediately reach the 
following 
corollary.

\begin{coro}\label{co1}
Given a circuit $C$ and its witness circuit $C'$. If the property $P$ is 
$k$-inductive in $C$, $P'$ is 1-inductive in $C'$.
\end{coro}

We continue to prove the reverse direction. Note that we are using the 
1-inductiveness of $C'$ in the following Lemma as an assumption instead 
of 
the BMC property as used in the Lemma~6 of \cite{DBLP:conf/cav/YuBH20}.

\begin{lemma}\label{lem1}
	If the property $P'$ is $1$-inductive in $C'$, then the BMC 
	check passes in $C$ for the unrolling 
	of length $k-1$.
\end{lemma}
\begin{proof}
	We assume $P'$ is 1-inductive such that the consecution check 
	($U'_1\wedge
	P'(I'_0,L'_0)\Rightarrow P'(I'_1,L'_1)$) and the BMC check 
	($R'(L_0)\Rightarrow P'(I'_0, L'_0)$) pass. We then do a proof by 
	contradiction 
	by
	assuming the BMC check in $C$ does not pass, meaning $R(L_0)\wedge
	U_{k-1}\wedge \neg \bigwedge\limits_{i\in[0,k)}P(I_i,L_i)$ has a satisfying
	assignment. We fix one such assignment and let $m\in [0,k)$ be the index for
	which $R(L_0)\wedge
	U_{k-1} \wedge \bigwedge\limits_{i\in[0,m)}P(I_i,L_i)\wedge\neg P(I_m,L_m)$ is
	satisfied. We consider two separate cases based on the value of $m$. 
	In the case where $m> 0$, we construct an assignment 
	such that $U'_1\wedge
	P'(I'_0,L'_0)\wedge\neg P'(I'_1,L'_1)$ is satisfied.
	Let $L_{0}^{k-m+i}=L_i, X_{0}^{k-m+i}=I_i, b_{0}^{k-m+i} = \top,  
	\text{ for } i \in [0,m), 
	b_{0}^{j} = \bot \text{ for } j \in [0, k-m)$. The assignment 
	of the initialisation bits satisfies $p_0$ as well as $p_4$. 
	By Def.~\ref{def-unrolling} we have $L_{i+1}\simeq F(I_{i}, L_{i})$ for all $i 
	\in[0,m)$, which satisfies the transition property. As 
	$\bigwedge\limits_{i\in[0,m)}P(I_i,L_i)$ is satisfied, $p_2$ is also satisfied. 
	Furthermore, $p_3$ is satisfied as $b^{k-m}_0$ is the oldest 
	initialisation 
	bit that is set, and $L_0^{k-m}=L_0$. The same assignment satisfies 
	$R(L_0)$.
	Therefore $P'(I_{0}', L_{0}')$ is satisfied. By 
	our assumption that the consecution check in $C'$ passes, we have 
	$P'(I_{1}', L_{1}')$. By 
	Def.~\ref{wit}, $L_1^{k-1}\simeq F(I'_0,L_0^{k-1})$, thus under our construction 
	$L_1^{k-1}=L_m.$ By Def.~\ref{wit}, 
	$P'(I_{1}', L_{1}')$ implies $P(I'_1, L^{k-1}_1)$.
	We let $I'_1=I_m$, and have $P(I_{m}, L_{m})$ thus the 
	contradiction 
	follows. We then consider the second case where $m=0$ such that 
	$R(L_0)\wedge U_{k-1}\wedge \neg P(I_0,L_0)$ is satisfied, where an 
	initial state in the original circuit is a bad state.  Based on $R(L_0)$ and 
	the fact that $R'$ is stratified, we can 
	construct a satisfying assignment for $R'(L'_0)$ with $L_0^{k-1}=L_0,
	I'_0=I_0$ and the rest uninitialised. Under our assumption that the BMC 
	check passes in 
	$C'$, we have $P'(I'_0,L'_0)$, and by $p_3$ in Def.~\ref{wit} and under 
	the variable 
	mapping, we have 
	$P(I_0,L_0)$. Therefore we reach a contradiction.
\end{proof}
\begin{lemma}\label{lem4}
If the consecution check for the unrolling of length $1$ passes in the 
witness circuit $C'$, 
then the 
consecution check also passes for the unrolling of length $k$ in 
$C$.
\end{lemma}
\begin{proof}            We assume $U'_1\wedge P(I'_0,L'_0)\Rightarrow
	P(I'_1,L'_1)$ holds. We then do a proof by contradiction by
	assuming that the 
	consecution check for the 
	original circuit fails. 
	Thus there is a 
	satisfying assignment $s$ of the formula		
	$U_k\wedge\bigwedge\limits_{i\in[0,k)} 
	P(I_i, L_i)\wedge \neg P(I_k,L_k)$. Based on $s$, we have a satisfying 
	assignment for 
	$U'_1\wedge P'(I'_0,L'_0)$ as follows. Let $X^i_0=I_i, L^i_0=L_i, 
	$ and $b^i_0=\top$ for $i\in[0,k)$. Let $X^{i-1}_1=I_i, L^{i-1}_1=L_i, 
	b^{i-1}_1=\top$ for $i\in[1,k]$. We now show this satisfies 
	$L'_1\simeq 
	F'(I'_0,L'_0)$. Since $X^{i-1}_1=I_i=X^{i}_0$ and $L^{i-1}_1=L_i=L^i_0$ for 
	all $i\in[1,k)$, 
	the transition relation of copying the components
	is satisfied. Since $s$ 
	satisfies $U_k$, by Def.~\ref{def-unrolling}, it satisfies $L_{k}\simeq 
	F(I_{k-1}, 
	L_{k-1})$. With $X^{k-1}_0=I_{k-1}, 
	L^{k-1}_0=L_{k-1}$, and $L^{k-1}_1=L_k$, we have $L^{k-1}_1 \simeq 
	F(X^{k-1}_0,L^{k-1}_0)$, which satisfies the transition definition for the most recent 
	copy. As for the initialisation bits, since all of 
	them are set to $\top$ in both $B_0$ and $B_1$, 
 the transition definition  of $B$ in Def.~\ref{wit} is satisfied. As a result, 
$U'_1$ is 
	satisfied, and we continue to show the same assignment satisfies 
	$P'(I'_0,L'_0)$. Similar to our proof in Lemma~\ref{lem2}, the values 
	of $B_0$ satisfy 
	$p_0(I_0',L'_0)$ and $p_4(I_0',L'_0)$ immediately. As the premiss of 
	$p_3(I'_0,L'_0)$ is unsatisfied, $p_3(I'_0,L'_0)$ trivially holds. Since 
	$U_k$ is satisfied, by 
	Def.~\ref{def-unrolling}, we have $L_{i+1}\simeq F(I_{i}, L_{i})$ for all 
	$i\in[0,k-1)$, thus also $p_1(I'_0,L'_0)$. Lastly, since 
	$P(I_i,L_i)$ is 
	satisfied 
	for all $i\in[0,k)$, the original property is satisfied in every 
	component 
	$P(X^i_0, L^i_0)$, resulting in the satisfaction of $p_2(I'_0,L'_0)$. By our 
	initial assumption, $P'(I'_1,L'_1)$ is satisfied. By 
	Theorem~\ref{thm-kwc-cs}, we have $P(X^{k-1}_1,L^{k-1}_1)$, thus 
	$P(I_k,L_k)$. We reach a contradiction here. We can 
	therefore conclude the consecution check of the original circuit 
	passes.
\end{proof}

We conclude the following corollary according to Lemma~\ref{lem1} and \ref{lem4}.

\begin{coro}~\label{co2}
	Given a circuit $C=(I, L, R, F, P)$ and its witness circuit $C'=(I', L', R', F', 
	P')$. If $P'$ is 1-inductive in $C'$, then $P$ is $k$-inductive in $C$.
\end{coro}

\begin{theorem}\label{mainthm}
	Given a circuit $C=(I, L, R, F, P)$ and its witness circuit $C'=(I', L', R', F', 
	P')$. $P$ is $k$-inductive in $C$ iff $P'$ is 1-inductive in $C'$.
\end{theorem}

Theorem~\ref{mainthm} follows Corollary~\ref{co1}  and~\ref{co2} directly.

\section{Implementation and Experimental Evaluation}~\label{exp}
We implemented the proposed certification approach into two complete 
toolkits\cite{Certifaiger}: \certifier for bit-level, and \wlcertifier for 
word-level.
We evaluate the performance of our tools 
against 
several benchmark sets from previous literature and the hardware model checking 
competitions and report the experimental results.

\subsection{Bit-level}\label{certifier}
We implemented our approach described in the previous section into an open 
toolkit called \certifier, which extends the certification toolkit 
\textsc{Certifaiger}~\cite{DBLP:conf/cav/YuBH20}. Our toolkit takes as inputs a 
model in AIGER~\cite{Biere-FMV-TR-11-2} format and a value of $k$ which 
is given by the model checker, then 
generates six SAT checks 
which are verified by a state-of-the-art SAT solver 
Kissat~\cite{BiereFazekasFleuryHeisinger-SAT-Competition-2020-solvers}. 
Note that the AIGER format only allows stratified resets by default.

We implemented a certificate generator based on the witness circuit 
construction as described in Def.~\ref{wit}. The generated witness circuit is 
passed to the certifier which consists of three parts which are independent 
of each other. The stratification 
check is implemented by the AIGER parser~\cite{Biere-FMV-TR-11-2}. The 
simulation 
checker generates three SAT checks as decribed in Def.~\ref{sim} 
(the reset check, the transition check, and the property check).  
The inductive 
invariant checker generates another three SAT checks (the initiation check, 
the consistency check, and the 
consecution check). Note that the inductive invariant checker is part of the 
general 
certification flow and is supposed to be not only for $k$-induction but also 
for various other model checking 
algorithms. Here we reuse the inductive invariant checker 
from \textsc{Certifaiger}. The certification 
is 
successful when all SAT checks are unsatisfiable.

For benchmarks we used the TIP suite from 
~\cite{DBLP:journals/entcs/EenS03} and benchmarks from the Hardware 
Model Checking Competition 2010~\cite{HWMCC10}. The performance of 
\certifier was evaluated 
against \textsc{Certifaiger} on both sets of benchmarks. All experiments 
were performed on a workstation with an 
Intel$^{\text{\textregistered}}$
Core$^{\textsc{tm}}$ i9-9900 CPU 3.60GHz computer with 32GB RAM 
running
Manjaro with Linux kernel 5.4.72-1. 

To determine the speedups of the new implementation proposed in this 
paper, we performed experiments on the same sets of the benchmarks 
used in ~\cite{DBLP:conf/cav/YuBH20}. The results are reported in Table 
~\ref{tb1}. There are significant overall gains in the initiation checks 
($\varphi_{init}$) as well 
as the reset checks ($\varphi_{reset}$). For the initiation check which 
checks the invariant holds in all initial states, the performance 
improvement is largely due to the 
simplification of the 
reset functions in
the new witness circuit construction. Interestingly, the SAT solving time for 
the new reset 
check is close 
to zero, which checks the equality of the reset functions between the 
shared set of latches and the latches in the original circuit. The 
explanation for such small execution time is that the 
reset functions for this set of benchmarks are $\bot$ for all latches, thus 
making the SAT checks rather trivial.

	\begin{table*}[t]
	
		\caption{Summary of certification results for the bit-level TIP suite. 
		\label{tb1}
	}
	\centering
	\begin{tabular}{l|rr|rr|rr|rr|rr|rr}
		
		\hline
		&  
		\multicolumn{2}{c|}{$\varphi_{init}$}    
		&\multicolumn{2}{c|}{$\varphi_{consist}$} &   
		\multicolumn{2}{c|}{$\varphi_{consec}$} 
		&   \multicolumn{2}{c|}{$\varphi_{trans}$} 
		&   \multicolumn{2}{c|}{$\varphi_{prop}$} 
		&   \multicolumn{2}{c}{$\varphi_{reset}$} \\

		Benchmarks & {$t_1$} & {$t_2$}  & {$t_1$} & {$t_2$}   & {$t_1$} & 
		{$t_2$}  & {$t_1$} & {$t_2$} 
		&  {$t_1$} & {$t_2$}  & {$t_1$} & {$t_2$}    \\
		\hline
		\hline
		c.periodic &7.78&0.06 & 0.06& 0.06   &56.82 & 
		56.29 &0.15&0.14
		& 0.05& 0.05 & 84.04&0.00    \\

		n.guidance$_1$&0.19& 0.01& 0.01&0.01 & 3.73 & 
		3.79 & 0.12& 0.12
		& 0.01& 0.01&1.21& 0.00   \\

		n.guidance$_7$& 4.09& 0.02  &0.02& 0.02 & 18.40& 
		18.17  &0.12&0.12
		&  0.02 & 0.02  & 25.22&0.00  \\

		n.tcasp$_2$&0.17 & 0.01 &0.01 &0.01  & 2.64 & 
		2.68 & 0.23 & 0.23
		& 0.01 & 0.02  &1.79 &0.00   \\

		n.tcasp$_3$ & 0.11& 0.01 & 0.01&0.01 & 1.82 & 
		1.70  &0.23& 0.26
		&  0.02 &0.02&1.01&0.00   \\
		
		v.prodcell$_{12}$ &2.35& 0.03 & 0.03&0.03& 59.05& 
		59.22 &0.12 &0.12
		&  0.03& 0.03  &8.48& 0.00   \\
		
		v.prodcell$_{13}$ &0.22& 0.01  & 0.01& 0.01   & 2.99 & 
		2.99  & 0.12&0.12
		&0.01&0.01 & 0.20 & 0.00  \\

		v.prodcell$_{14}$ &0.64& 0.02 & 0.02&0.02  &13.69 & 
		13.69  &0.12& 0.12
		&  0.02 &0.02 & 1.45& 0.00   \\

		v.prodcell$_{15}$ &2.22& 0.02& 0.03& 0.03   &32.66 & 
	32.28  & 0.12&0.12
		&  0.02& 0.02 & 2.26&0.00 \\

		v.prodcell$_{16}$ &0.01&0.01  & 0.01 & 0.01   &1.19 & 
		1.20&0.12 & 0.12
		&  0.01 & 0.01  &0.06& 0.00 \\
		v.prodcell$_{17}$ &2.34&0.03 &0.03 & 0.03 & 48.51& 
		48.17  & 0.12 & 0.12
		& 0.03&0.03  & 6.86& 0.00   \\
		v.prodcell$_{18}$ &0.67& 0.01 & 0.01& 0.01  & 8.67 & 
		8.78 &0.12 &0.12
		& 0.02& 0.02 & 0.79 & 0.00 \\
			v.prodcell$_{19}$ & 1.66 & 0.02&0.02& 0.03 & 31.98& 
	31.78 & 0.12& 0.12
		& 0.03 &0.03 & 3.73& 0.00  \\
		v.prodcell$_{24}$ & 3.32&0.04  &0.04 & 0.04 & 112.12& 
		115.18& 0.12 & 0.12&0.04 & 0.04
		& 17.64 & 0.00     \\
		
		\hline
		\hline

		\hline
	\end{tabular}
	
\vspace{2mm}
	Columns 
	report the 
	benchmark names, and the time (in seconds) used for each SAT check by 
	\textsc{certifaiger} ($t_1$) 
	and \certifier ($t_2$) respectively.
\end{table*}
\begin{figure*}[htp] \label{hw10}
	\centering
	
	\includegraphics[width=\textwidth]{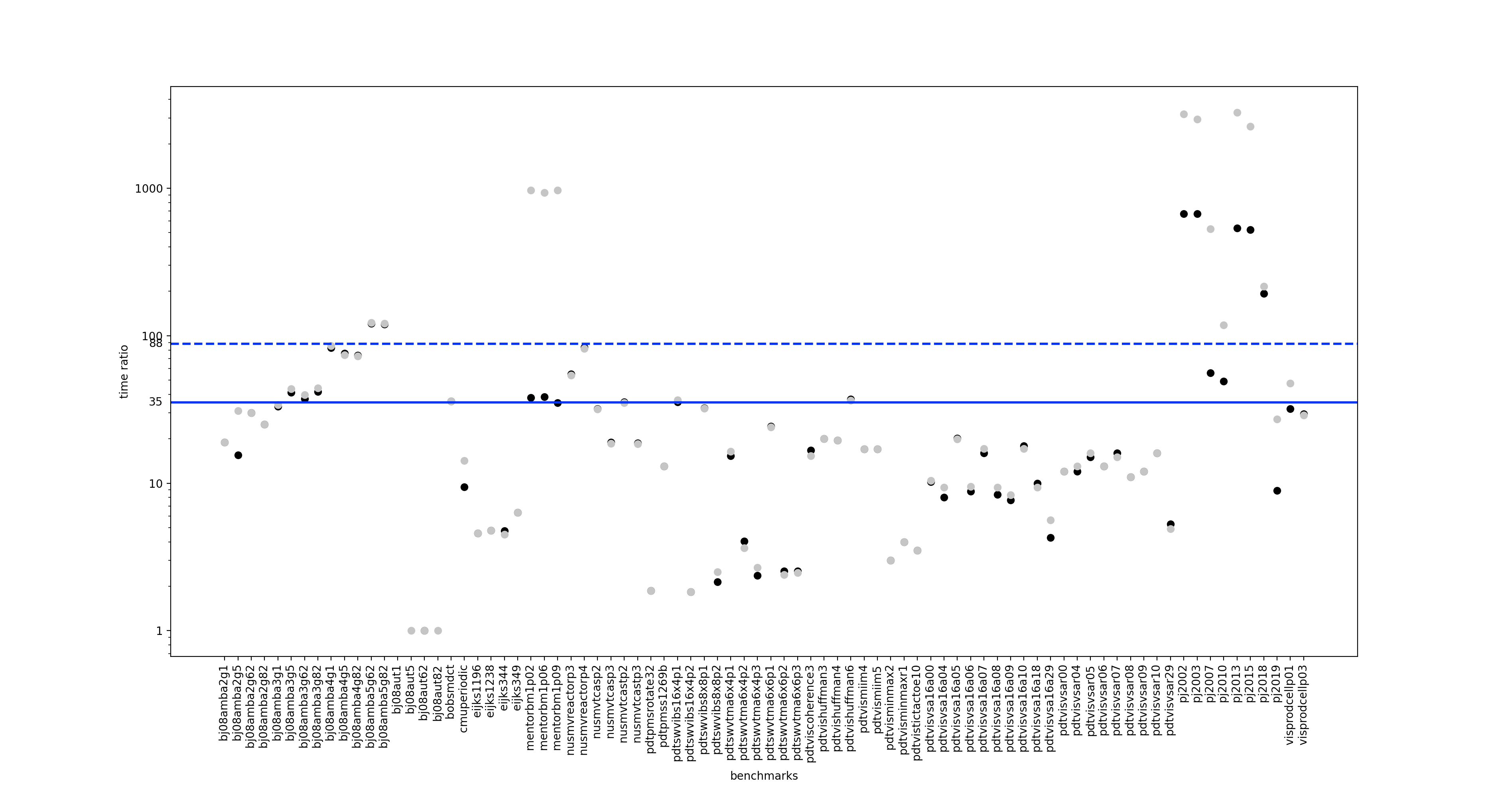}%

	\caption{Bit-level: the experimental results of the HWMCC 2010. The 
		benchmark names are shown on the x-axis. The time ratio on the 
		y-axis is calculated by computing certification time divided by model 
		checking time. The black dots in the graph are the 
		results obtained from \certifier and the grey dots are from 
		\textsc{Certifaiger}. The straight line and
		the dashed line are the calculated means for 
		\certifier and \textsc{Certifaiger} respectively. \label{fig:hwmcc}}
\end{figure*}

The results in Fig.~\ref{fig:hwmcc} demonstrate that \certifier in general 
is much faster than \textsc{certifaiger} during the overall certification 
process. We plotted the time ratios of the total certification time and the 
model checking time (ran on the model checker 
McAiger~\cite{DBLP:conf/fmcad/BiereB08}) 
against the HWMCC benchmarks. Compared to \textsc{Certifaiger}, 
\certifier achieved overall speedups of 2.46 times. 
As we can see from 
the plot, 
especially for the instances with certification time greater than 500 
seconds, the new implementation significantly improved 
the certification performance. We observe performance gains in most 
benchmarks, 
as the previous performance bottleneck for certain benchmarks is the 
QBF solving time for the reset check. For other instances, the bottleneck 
is the 
SAT solving time for the consecution check, which is also improved due to 
a simpler reset construction (as part of the inductive invariant).

\subsection{Word-level}
We further lifted the method to certifying word-level model checking by 
implementing an experimental toolkit called \wlcertifier written in Nim. 
\wlcertifier follows 
the same architecture design as \certifier and uses Boolector 
~\cite{DBLP:conf/cav/NiemetzPWB18} as the underlying SMT solver. All 
models and SMT encodings are in 
BTOR2~\cite{DBLP:conf/cav/NiemetzPWB18} format, which 
is the standard word-level model checking format used in hardware model 
checking 
competitions. Upon invocation the tool first generates a witness circuit 
according to the given value of $k$; then six SMT formulas from 
the inductive invariant check and the stratified simulation check as 
described in Section~\ref{certification}. The resulting SMT formulas are 
verified by Boolector. The stratification check is implemented in the 
BTOR2 parser. The certification passes if all six 
SMT formulas are proven unsatisfiable. In the case of $k=1$, no witness 
circuit is generated and the inductive invariant checker is executed on the 
original 
model.

In order to evaluate our implementation and its scalability, we consider 
two sets of experiments. First of all we consider the Counter example, 
in which there is a 500-bit counter doing exact-modulo 
checks with 
initial 
value 0, and an 
enabler 
which resets the counter value to 0. The counter increments by 1 at each 
timestamp until it reaches  a 
modulo bound $m$ (then resets). The property encodes the fact that the 
counter 
value does not 
reach some bad state $b$, which is then 
$k$-inductive with $k = b-m+1$.
For the experiments, we fixed
the modulo bound at 32 and scaled the inductive depth up to 1000.

We ran all benchmarks on 
the word-level model checker 
AVR~\cite{DBLP:conf/tacas/GoelS20} to get the 
values of $k$. Fig.~\ref{fig:counter} shows the experimental results 
obtained with 
\wlcertifier 
under the same setting as Section~\ref{certifier}.  We record model checking 
time and certification time separately to show comparison. Interestingly, 
the certification time is much lower than the model checking time as can be 
seen in the diagram, meaning certification is at a low cost. As the value of 
$k$ increases, on average the certification time is proportionally lower. 

In Table ~\ref{tb2}  we report the 
experimental results obtained on a superset of the hardware model 
checking competition 2020~\cite{hwmcc20} 
benchmarks\footnote[1]{All benchmarks will be published in the artifact.}. 
To select the benchmarks presented, we first ran AVR with a
timeout of 5000 seconds. We display the results in Table~\ref{tb2} that 
are of particular interest with a running time of more than ten seconds (there
are 7 instances with $k=1$ which were certified and solved under 0.2s). We 
observe that the certification time is much lower than model checking time.
Including certification would increase the runtime of AVR on the model checking
benchmarks by less than 6\%.
As expected, the consecution check as the most complicated formula
among the six takes up the majority of the running time.  

\begin{table*}[t]

	\caption{Summary of certification results word-level benchmarks from 
	the HWMCC20
		\label{tb2}
}
\centering
	
\begin{tabular}{lrrrSSSS}
\hline
 Benchmarks                           &   {k} &   {\#model} &   {\#witness} &      {ModelCh.} &   {Certifi.} &   {Consec.} &   {Ratio} \\
\hline
\hline
 paper\_v3                             & 256 &       35 &      12801 &      10.25 &       1.14 &      0.90 &    0.11 \\
 VexRiscv-regch0-15-p0                &  17 &     2149 &      43077 &      10.31 &       4.04 &      3.29 &    0.39 \\
 zipcpu-pfcache-p02                   &  37 &     1818 &     105874 &      13.95 &       4.40 &      2.73 &    0.32 \\
 zipcpu-pfcache-p24                   &  37 &     1818 &     105874 &      14.35 &       4.49 &      2.83 &    0.31 \\
 zipcpu-busdelay-p43                  & 101 &      950 &     145466 &      15.29 &       6.14 &      3.86 &    0.40 \\
 dspfilters\_fastfir\_second-p42        &  15 &     6732 &     115388 &      16.11 &      14.80 &     12.96 &    0.92 \\
 zipcpu-pfcache-p01                   &  41 &     1818 &     117434 &      18.33 &       6.34 &      4.47 &    0.35 \\
 dspfilters\_fastfir\_second-p10        &  11 &     6732 &      84348 &      24.56 &       9.76 &      8.44 &    0.40 \\
 zipcpu-busdelay-p15                  & 101 &      950 &     145466 &      58.17 &       8.18 &      5.89 &    0.14 \\
 qspiflash\_dualflexpress\_divfive-p120 &  97 &     3100 &     394412 &      63.58 &      22.07 &     14.58 &    0.35 \\
 zipcpu-pfcache-p22                   &  93 &     1818 &     267714 &     166.07 &      23.66 &     19.06 &    0.14 \\
 VexRiscv-regch0-20-p0                &  22 &     2149 &      55862 &     240.50 &      16.76 &     15.76 &    0.07 \\
 dspfilters\_fastfir\_second-p14        &  15 &     6732 &     115388 &     354.01 &      21.27 &     19.44 &    0.06 \\
 dspfilters\_fastfir\_second-p11        &  21 &     6732 &     161948 &     627.69 &      46.88 &     44.30 &    0.07 \\
 dspfilters\_fastfir\_second-p45        &  17 &     6732 &     130908 &    1094.11 &      30.14 &     28.06 &    0.03 \\
 VexRiscv-regch0-30-p1                &  32 &     2150 &      81464 &    1444.47 &      83.38 &     81.95 &    0.06 \\
 dspfilters\_fastfir\_second-p43        &  19 &     6732 &     146428 &    2813.61 &      58.02 &     55.69 &    0.02 \\
\hline
\hline

\hline
\end{tabular}

\vspace{2mm}
	Columns
	report the
	benchmark names,
	the value of $k$,
	the size of the model (measured in number of instructions) and the 
	generated witness,
	 the model checking time, and certification time (in seconds).
	Additionally we list the time Boolector took to solve the consecution 
	check, as well as the ratio of model checking vs. certification 
	time. We only list the consecution check (Consec.) here as it takes up the 
	majority of the certification time.  
\end{table*}

\begin{figure*}[htp]
	\centering

	\includegraphics[width=0.97\textwidth]{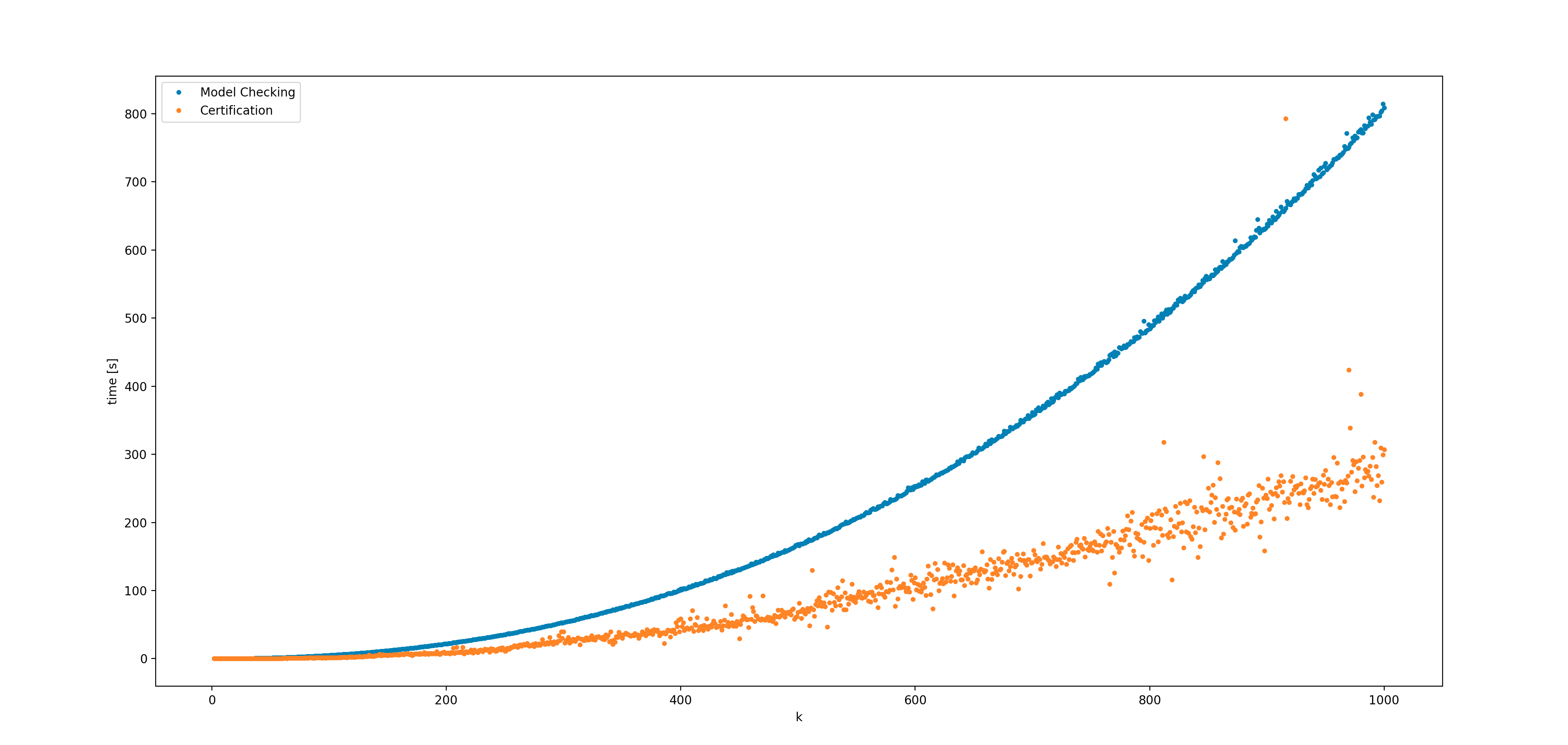}%
	
	\caption{\label{fig:counter}Word-level: model checking vs. certification 
	time for the 
	Counter example 
	with increasing values of $k$. The certification time is significantly 
	smaller than the model checking time.}
\end{figure*}

\section{Conclusion and future work}
We have presented a new certification framework which allows certification 
for $k$-induction to be done by six SAT checks and a polynomial-time 
check. We further lifted our 
approach from bit-level to word-level, and implemented our method in 
both contexts. Experimental results on various sets of benchmarks 
demonstrate the effectiveness and computational efficiency of our toolkits. 
Time needed for word-level certification is significantly smaller than the 
time needed for word-level model checking. This makes adding 
certification for word-level model checking a very interesting and viable 
proposition. The removal of the QBF quantifiers has reduced the theoretical 
complexity of the problem compared to ~\cite{DBLP:conf/cav/YuBH20} and 
also reduced the overall runtime overhead of the certification.

We leave to
future work how to certify $k$-induction with simple paths constraints.
Additionally, we plan to obtain formally verified certificate checkers by using theorem proving.
Finally, how to certify liveness properties is another important avenue of further research.

\subsubsection*{Acknowledgements}

This work is supported by the Austrian Science Fund (FWF) under the project
W1255-N23, the LIT AI Lab funded by the State of Upper
Austria, and Academy of Finland under the project 336092.

\bibliographystyle{IEEEtran}
\bibliography{paper}
\end{document}